\documentclass[11pt,a4paper,reqno]{amsart}
\pdfoutput=1
\usepackage[foot]{amsaddr}
\usepackage{latexsym}
\usepackage[utf8]{inputenc}
\usepackage[english]{babel}
\usepackage{amsmath}
\usepackage{amsfonts}
\usepackage{amssymb}
\usepackage{array}
\usepackage{enumitem}
\usepackage{graphics, graphicx}
\usepackage{bmpsize}
\usepackage{pifont}
\usepackage[ruled,vlined]{algorithm2e}
\usepackage{epsfig}
\usepackage{rotating}
\usepackage{natbib}
\setcitestyle{authoryear, comma}
\usepackage{underscore}
\usepackage{epstopdf}
\usepackage{eso-pic}
\usepackage[toc,page]{appendix}
\usepackage{multicol, multirow, bm}
\usepackage{url}
\usepackage{float}
\usepackage{paracol}
\newcommand{\defeq}{\mathrel{\mathop:}=}
\usepackage{afterpage}

\usepackage{setspace}
\usepackage{caption}
\usepackage{fullpage}
\usepackage{amsthm}
\usepackage{cancel}
\usepackage{nicefrac, longtable}
\usepackage{parskip}

\usepackage{mathtools}

\DeclarePairedDelimiter{\floor}{\lfloor}{\rfloor}
\usepackage{hyperref}
\hypersetup{
citecolor=black,%
filecolor=black,%
linkcolor=black,%
urlcolor=black
}

\newtheorem{theorem}{Theorem}[section]

\newtheorem{proposition}[theorem]{Proposition}

\newtheorem{remark}[theorem]{Remark}

\author[C.~Damian]{Camilla Damian$^\ast$}
\address{$^\ast$Institute of Statistics and Mathematical Methods in Economics, TU Wien}
\email{camilla.damian@tuwien.ac.at}
\author[R.~Frey]{R\"{u}diger Frey$^\dagger$}
\address{$^\dagger$Institute for Statistics and Mathematics, Vienna University of Economics and Business}
\email{ruediger.frey@wu.ac.at}

\date{\today}

\thanks{The first author was in part supported by the Austrian Science Fund (FWF, grant ZK 35 and grant Y 1235).}

\begin{document}

\begin{abstract}
In this paper, we focus on the estimation of historical volatility of asset prices from high-frequency data. Stochastic volatility models pose a major statistical challenge: since in reality historical volatility is not observable, its current level and, possibly, the parameters governing its dynamics have to be estimated from the observable time series of asset prices. To complicate matters further, recent research has analyzed the rough behavior of volatility time series to challenge the common assumption that the volatility process is a Brownian semimartingale. In order to tackle the arising inferential task efficiently in this setting, we use the fact that a fractional Brownian motion can be represented as a superposition of Markovian semimartingales (Ornstein-Uhlenbeck processes) and we solve the filtering (and parameter estimation) problem by resorting to more `standard' techniques, such as particle methods.
\end{abstract}

\title[]{Detecting Rough Volatility: A Filtering Approach}

\maketitle
\vspace{-0.5cm}
\begin{center}
\footnotesize \textbf{Keywords}: high-frequency data; rough volatility; nested particle filter.
\end{center}
\vspace{0.5cm}

\section{Introduction}\label{sec:IntrfBM}
It is well-known that, from a statistical point of view, one of the main challenges in the context of \textit{stochastic volatility} models is that of inference: since in reality historical volatility is not observable, its current level and, possibly, the parameters governing its dynamics have to be estimated from the observable time series of asset prices. In particular, the unobservability of volatility represents a crucial issue when modelling high-frequency data, i.e. data recorded on an irregular and remarkably small time scale.

To complicate matters further, recent research has analyzed the \textit{rough} behavior of volatility time series to challenge the assumption, common to most stochastic volatility models, that the volatility process is a Brownian semimartingale. For instance, \cite{volrough} use historical volatility proxy data to estimate the smoothness of the volatility process with a regression approach. Their empirical finding is that increments of log-volatility of several assets enjoy a monofractal scaling property: this, together with the well-established stylized fact that the distribution of such increments is approximately Gaussian, means log-volatility behaves essentially as a fractional Brownian motion. In particular, given their estimates of the smoothness parameter, these authors claim that log-volatility can be modeled as a fractional Brownian motion with Hurst exponent $H \approx 0.1$. This is in contrast to other fractional stochastic volatility models in the literature (see \cite{comte1998long}) -- typically assuming an Hurst index $H > \nicefrac{1}{2}$ -- and, more generally, to the quite widespread belief that volatility is a long-memory process.

A related work is that of \cite{bennedsen2021decoupling}, who focus on intraday volatility and, by means of an extensive empirical study, find it to be rough \textit{and} persistent. However, as explained in \cite{gneiting2004stochastic}, models based on processes with the self-similarity property, such as fractional Brownian motion, cannot account for both effects. Therefore, \cite{bennedsen2021decoupling} propose to use the so-called Brownian semistationary processes (\cite{barndorff2009brownian}) to model volatility instead, as this choice allows to decouple the short- and the long-term behavior of volatility. If we focus on their estimation of the roughness parameter, we find similarities with \cite{volrough} both in the linear regression approach\footnote{Although \cite{bennedsen2021decoupling} also present an estimation procedure based on a non-linear regression: this alternative can mitigate the bias arising from using a volatility proxy, which can affect the OLS estimate of the roughness parameter. However, in practice, the two estimators behave in a qualitatively similar fashion.} and in the results. In fact, the estimates for $H$, computed for a variety of assets, are qualitatively in agreement with those found by \cite{volrough}. 

\cite{fukasawa2019volatility}, however, apply the regression methodology of \cite{volrough} using the 5-minute realized volatility from simulated price paths and find that, given this sampling frequency, this approach gives a remarkable fit with $H \approx 0.1$ \textit{regardless} of the true value of $H$ used in their simulation. They attribute this finding to the use of a volatility proxy, claiming that it is the error in approximating spot volatility that results in an illusory scaling property. Prompted by these findings, these authors take such `proxy error' explicitly into account and develop a novel quasi-likelihood estimator of the Hurst exponent. In any case, when applying their estimator to real data, \cite{fukasawa2019volatility} find that volatility is indeed rough: remarkably, their estimates for $H$ for various stock indices are not only smaller than $0.5$, but even smaller than those of \cite{volrough}. On the other hand, further follow-up studies to \cite{volrough}, such as \cite{rogers2023things} and -- more recently -- \cite{cont2022rough}, find that `simpler' (meaning driven by a Brownian diffusion) volatility models can explain market features just as well as fractional ones. \cite{cont2022rough}, in particular, identify the cause of apparent roughness in realized volatility with microstructure noise rather than with an inherent, `true' roughness in (spot) volatility.

In this paper, the recent literature regarding \textit{rough} volatility, as well as its intrinsic unobservability, will be taken into account. As we will explain in the next section, we assume that the asset price process changes at discrete times (essentially, the times when new orders arrive to the market), which we model as jump times of a doubly-stochastic Poisson process. We take the intensity of the jump process to be driven by an unobservable signal and we explain how such intensity relates to asset price volatility. To account for roughness in volatility, the central idea is to consider the dynamics of the unobservable process to be driven by a fractional Brownian motion (fBM) with a Hurst index $H < \nicefrac{1}{2}$. This results in the volatility process being neither a semimartingale nor a Markov process, posing considerable challenges in the analysis of the model, in particular for what concerns the filtering and parameter estimation tasks. Even in the flexible context of a particle filtering approach, the non-Markovianity of the signal implies that sampling a particle at a given time instance requires calculations involving the entire trajectory up to that point. Clearly, the practical consequence is a great increase in computational cost and in memory allocation requirements over time.

Here, we circumvent this issue by employing a suitable approximation of the underlying non-Markovian volatility process. The starting point is the work of \cite{carmona1998fractional} and \cite{carmona2000approximation}: they have shown that one can represent an fBM with $H < \nicefrac{1}{2}$ as a superposition of infinitely many Ornstein-Uhlenbeck processes. We can exploit this idea to approximate the unobservable state process by means of a finite number of such processes, each of them being a Markovian semimartingale: this allows us to use a standard (recursive) particle filter to retrieve the signal; moreover, we can infer about model parameters -- in particular, the Hurst index -- by adapting the \textit{nested particle filter algorithm} of \cite{crisan2018nested} to our specific framework. In this way, we can reformulate the complex problem of determining the roughness of a hidden path in terms of a -- more easily approachable -- parameter estimation task. Simulation experiments indicate that this approach yields satisfactory results, both in terms of filtering and of parameter estimation, while still fully taking into account the intrinsic unobservability of the volatility process. Moreover, in the spirit of \cite{rogers2023things} and \cite{cont2022rough}, we apply our methodology to synthetic data generated in a setting where (spot) volatility is driven by a Brownian diffusion and assess whether the resulting estimates for $H$ are more consistent with such a setup, rather than mistakenly reflecting a spurious strong roughness effect.

The remainder of this paper is as follows. Section~\ref{section:MOD} introduces the model: motivated by a specific continuous-time modelling framework involving a fractional resp. Liouville Brownian motion with $H < \nicefrac{1}{2}$, we explain how such framework can be discretized in time and how such processes can be approximated to obtain a modelling framework suitable for the algorithm of \cite{crisan2018nested}, which is presented in Section~\ref{section:PF}. The numerical experiments are presented in Section~\ref{section:NUM}; in particular, we include an application of the proposed methodology to the ``alternative'' models of and \cite{cont2022rough} and \cite{rogers2023things}. Section~\ref{section:concl} concludes and illustrates possible extensions.

\section{Model}\label{section:MOD}
In this section, we discuss key ingredients of our model in detail. We start by explaining the motivation behind our specific model design choices for observation and state process in continuous time; then, we proceed to illustrate a suitable approximation of the underlying non-Markovian state process that will make it possible to tackle filtering and parameter estimation in our context. We conclude this section with the description of the discrete counterpart of our modelling framework, which we will need in order to apply the algorithm in Section~\ref{section:PF}.
\subsection{Motivation and Continuous-Time Framework}\label{sub:CTFram}
As mentioned in the introduction, our goal is to investigate the tasks of volatility estimation and inference in the context of a model which is suitable for high-frequency data and which accounts for \textit{roughness} and \textit{unobservability} of volatility dynamics. We can include these features in our setup by means of specific design choices concerning \textit{observation} and \textit{state} processes, respectively.

\subsubsection{Observations}
The inspiration for the modelling framework considered in this paper lies in \cite{freyrungg2001}, who also consider the problem of volatility estimation when observations consist of high-frequency data. In their model, a (logarithmic) asset price is assumed to change -- and to be observed -- only at random, discrete points in time in order to better mimic the characteristics of high-frequency data. These random times are considered to be the jump times of a marked point process whose intensity and jump-size distribution depend on the level of a hidden state process, closely related to asset price volatility; thus, volatility estimation amounts to a filtering problem with marked point process observations. This represents a departure from more commonly used diffusion models, such as the Heston model, and it makes for a more plausible modelling framework for asset prices observed on a very small time-scale. In particular, the dependence of the intensity on a hidden state variable introduces randomness in market activity over time.

Therefore, with the application to high-frequency observations in mind, we follow \cite{freyrungg2001} in that we model observed prices with a jump process, rather than with a continuous one. However, note that in this paper the plan is to consider \textit{event} data (that is, all transactions). As prices recorded at such ultra high frequency often move only by one or few ticks, the size of the jumps is rather uninformative; thus, we can simplify the filtering task by assuming that the unobservable signal influences only the intensity of the point process modelling the observed jump times. We will specify the form of the intensity in more detail in Equation~\eqref{eq:int} (continuous-time framework) resp. \eqref{eq:pcint} (discrete-time framework); moreover, we will explain the relationship between our model and a continuous model with stochastic volatility.

Another important difference between our setup and the original one from \cite{freyrungg2001} is that they assume the hidden state process to be a time-homogeneous Markov process, a setting which is not suitable for fractional or Liouville Brownian motion. Following the approach of \cite{kushner1977probability}, they obtain a recursive approximation to the optimal filter. In \cite{cvitanic2006filtering} these results are extended and, using a Bayesian approach, a recursive, closed-form optimal filter is obtained. Introducing \textit{rough} volatility amounts to considering the state process dynamics to be driven by a non-Markovian, non-semimartingale process, which prevents us from exploiting these results directly. The details on the specific form of the state process and on how to tackle the corresponding filtering problem will be the topic of the following sections.

\subsubsection{Hidden State}
As mentioned previously, it is crucial in this paper to choose a model which can account for volatility \textit{roughness}, in alignment with results in recent literature. In our framework, this can be achieved by considering the dynamics of the hidden state process to be driven by a fractional Brownian motion with Hurst index $H < \nicefrac{1}{2}$, or by a process `close' to it.

Note that, in such a rough framework, it is not possible to use `classical' filtering results directly because the state process is neither a semimartingale nor a Markov process. Some theoretical results regarding nonlinear filtering for fractional Brownian motion (both in the state and in the observation process) have been obtained for example by \cite{decreusefond1998fractional} and \cite{coutin1999abstract}, but they are restricted to the case $H \geq \nicefrac{1}{2}$ and thus not appropriate for the \textit{rough} volatility setting considered here. However, this issue can be bypassed by employing a suitable approximation of the underlying non-Markovian volatility process, as we will describe in detail in Section~\ref{sub:approx}.

\subsubsection{Continuous-Time Modelling Framework}
Now that we have introduced the main ingredients of the model, we can describe the corresponding framework. We will generically denote by $X$ the hidden state process, defined on some underlying filtered probability space  $(\Omega, \mathcal{F}, \mathbb{F}, \mathbb{P})$, where $\mathbb{F} = (\mathcal{F}_t )_{0 \leq t \leq T}$ satisfies the usual conditions. Later, in Section~\ref{sub:approx}, we will explain how $X$ is constructed to approximate a fractional resp. Liouville Brownian motion with Hurst index $H < \nicefrac{1}{2}$; in particular, we will see how $X$ is in fact a functional of a finite number of Ornstein-Uhlenbeck process driven by the same Brownian motion.

Then, in continuous time, our observation process is a Cox process denoted by $D$ for which we assume an intensity of the form
\begin{equation}\label{eq:int}
\lambda_t \defeq \lambda(X_t) = b \cdot \exp\left(X_t\right) \, ,
\end{equation}
where $b$ is a positive constant.

\subsubsection{Relationship to Stochastic Volatility}
Now that the continuous-time model has been introduced, we shall detail the relationship between the intensity of $D$ and stochastic volatility. To this, we assume that a stock price $S$ is given by the following doubly-stochastic compound Poisson process
\begin{equation}\label{eq:stock}
S_t = S_0 + \sum_{i = 1}^{D_t} \nu_i, \quad \{\nu_i\}_{i \in \mathbb{N}} \, \text{i.i.d.}, \quad \mathbb{E}\{\nu_i\} = 0, \quad \text{Var}(\nu_i) = \sigma^2 \, ,
\end{equation}
where $\sigma > 0$ and $D$ is a doubly-stochastic Poisson process with intensity $\lambda(X_t) = b \cdot \exp\left(X_t\right)$, for a Gaussian process $X$. Moreover, $D$ and $\{\nu_i\}_{i \in \mathbb{N}}$ are independent.

Note that, when the observations consist of event data, the $\{\nu_i\}_{i \in \mathbb{N}}$ are small and bounded as prices move only by one or few ticks at a time; at the same time, there are many jumps (that is, there are frequent trades), so that $b > 0$ should be large. Hence, one expects that the model~\eqref{eq:stock} is close to a stochastic volatility model with instantaneous variance given by $\sigma_t^2 = \sigma^2 \lambda(X_t)$. This concept is formalized in Proposition~\ref{prop:intvol} in Appendix~\ref{app}.

Now we shall detail how $X$ can be chosen such that the modelling framework accounts for roughness in volatility dynamics, while being fairly tractable in the sense that the filtering and parameter estimation tasks for $X$ can be performed via more `standard' methods.

\subsection{Approximation of fBM}\label{sub:approx}
In this section, we will explain how we can use some results from \cite{carmona1998fractional} and \cite{carmona2000approximation} to approximate fractional Brownian motion (resp. Liouville Brownian motion), as well as how this is helpful in our context.

More specifically, one considers the Mandelbrot-Van Ness representation of a fractional Brownian motion $W^H$, which is given by
\begin{equation}\label{eq:MVN}
W^H_t = c_H \int_0^t (t - s)^{H - \frac{1}{2}} \, dB_s
+ c_H \int_{-\infty}^0 \left((t - s)^{H - \frac{1}{2}} - (-s)\right)^{H - \frac{1}{2}}\, dB_s
\end{equation}
where $B$ is a two-sided Brownian motion and $c_H$ a constant depending on $H$. In particular, the choice
\begin{equation}\label{eq:cH}
c_H = \sqrt{\frac{\pi H (2H - 1)}{\Gamma(2 - 2H)\Gamma(H + \nicefrac{1}{2})^2 \sin(\pi(H - \nicefrac{1}{2}))}}
\end{equation}
ensures that the autocovariance function of $W^H$ is given by 
\begin{equation}\label{eq:autocovf}
\mathbb{E}\{W_t^H W_s^H\} = \frac{1}{2} \left\{|t|^{2H} + |s|^{2H} - |t - s|^{2H}\right\} \, .
\end{equation}

Similarly, here we introduce the so-called Liouville Brownian motion $V^H$ as
\begin{equation}\label{eq:LiouBM}
V^H_t = c_H \int_0^t (t - s)^{H - \frac{1}{2}} \, dB_s \, .
\end{equation}

Now, if $H < \nicefrac{1}{2}$, it has been shown in \cite{carmona1998fractional} and \cite{carmona2000approximation} that expressing $u \mapsto u^{H - \frac{1}{2}}$ as a Laplace transform
\[
u^{H - \frac{1}{2}} \propto \int_0^\infty  e^{-xu} x^{-H - \frac{1}{2}} \, dx
\]
and applying stochastic Fubini theorem results in a Markovian representation. In particular, we can express \eqref{eq:MVN} as
\[
W^H_t = \int_0^\infty \int_0^t e^{-x(t - s)} \, dB_s \, \mu(dx) + \int_0^\infty \left(e^{-xt} - 1\right) \int_{-\infty}^0 e^{sx}  \, dB_s \, \mu(dx) \, ,
\]
where
\begin{equation}\label{eq:mu}
\mu(dx) = c_H \frac{x^{-H - \nicefrac{1}{2}}}{\Gamma(\nicefrac{1}{2} - H)} \, dx.
\end{equation}

Denote by $Z^x$ the Ornstein-Uhlenbeck process $\int_0^t e^{-x(t - s)} \, dB_s$, starting at zero, and set $Q_0^x \defeq \int_{- \infty}^0 e^{xs} \, dB_s$, so that
\[
W_t^H = \int_0^\infty Z_t^x \, \mu(dx) + \int_0^\infty \left(e^{-xt} - 1\right) Q_0^x \, \mu(dx)
\]
and
\[
V_t^H = \int_0^\infty Z_t^x \, \mu(dx) \, .
\]

\subsubsection*{Spatial Discretization} This idea can be exploited in the context of this paper as follows. For a fixed Hurst parameter $H < \nicefrac{1}{2}$, we will approximate the corresponding fractional resp. Liouville Brownian motion by a finite sum of Ornstein-Uhlenbeck processes (all driven by the same Brownian motion). This, in turn, will allow us to perform filtering for the unobservable state process by employing more `standard' techniques, such as particle filtering. The core idea is to approximate the measure $\mu$ in \eqref{eq:mu} by a finite sum of Dirac measures; that is, for some $J \in \mathbb{N}$, $\mu \approx \sum_{j = 1}^{J} c_j \delta_{\kappa_j}$ for positive coefficients $(c_j)_{j = 1, \dots, J}$ and positive mean-reversion speeds $(\kappa_j)_{j = 1, \dots, J}$. 

In particular, similarly to \cite{carmona1998fractional} and \cite{carmona2000approximation}, we perform the following spatial discretization. Given $H \in (0, \nicefrac{1}{2})$, fix $J \in \mathbb{N}$ and consider $[\xi_0, \xi_J]$, a compact subset of $(0, \infty)$. Split this interval into subintervals by auxiliary terms $0 < \xi_0 < \xi_1 < \dots < \xi_J < \infty$ and compute, for $j = 1, \dots, J$,
\[
c_j = \int_{\xi_{j - 1}}^{\xi_j} \mu(dx) \quad \quad \text{and} \quad \quad \kappa_j = \frac{1}{c_j} \int_{\xi_{j - 1}}^{\xi_j} x\mu(dx) \, .
\]
With a slight abuse of notation, now and in what follows we write $Z^j$ and $Q_0^j$ in place of $Z^{\kappa_j}$ resp. $Q_0^{\kappa_j}$.

An approximation of fractional Brownian motion is then given by:
\begin{equation}\label{eq:finapproxWH}
W_t^H \approx \sum_{j = 1}^J c_j \left(Z_t^j + \left(e^{-\kappa_j t} - 1\right) Q_0^j\right) \,,
\end{equation}
where $\left(Q_0^j\right)_{j = 1}^J$ can be easily simulated, since $\left(Q_0^x\right)_{x > 0}$ is a centered Gaussian process with covariance function $\Gamma(x, y) = \frac{1}{x + y}$.

In this paper, we will instead consider the following approximation of the Liouville Brownian motion $V^H$ of Equation~\eqref{eq:LiouBM}. This approximation is denoted by $X$, that is
\begin{equation}\label{eq:finapprox}
V_t^H \approx X_t \defeq \sum_{j = 1}^J c_j Z_t^j \,,
\end{equation}
as this eases the exposition and notation in what follows.
\subsubsection*{Quality of the Approximation} Several papers study the accuracy and the convergence of the approximation of fractional Brownian motion via a finite number of OU processes: other than \cite{carmona1998fractional} and \cite{carmona2000approximation}, we refer the readers to \cite{harms2019strong} for an approach involving quadrature rules and to \cite{coutin2007approximation} for further results concerning also the temporal approximation.

A standard choice for auxiliary terms is a geometric partition of ratio $r \in (1, 2)$. The quality of the approximation obviously depends on the chosen $J$ and $r$. Numerical experiments, in agreement with theoretical results, suggest that for large $J$ one should choose a value of $r$ close to $1$. However, if in practical applications we choose to keep the dimension of the approximation manageable, numerical experiments indicate that we should increase the value of $r$. In particular, for a given natural number $J > 16$, here we fix a compact $[\xi_0, \xi_J]$ and split it into subintervals using a geometric partition, with ratio $\left(\frac{\xi_J}{\xi_0}\right)^{\nicefrac{1}{J}}$. Typically, one would want to choose a smaller value of $\xi_0$ for $H$ close to $0.5$, and a larger value of $\xi_J$ for $H$ close to $0$. Writing $\alpha = H + \nicefrac{1}{2}$, we can, for example, choose the values $\xi_0 = J^{-2\alpha}$ and $\xi_J = J^{4 - 2\alpha}$, so that $\left(\frac{\xi_J}{\xi_0}\right)^{\nicefrac{1}{J}} = J^{\nicefrac{4}{J}}$. Thus, with careful choices for $\xi_0$ and $\xi_J$, we can fix $J$ and construct a partition which is in line with what observed above and performs reasonably well numerically also for modest $J > 16$.

\begin{remark}\label{rmk:rmk2}
In practice, note that the choice of $J$ will also depend on the fineness of the time scale on which we want to approximate $W^H$ resp. $V^H$. Fix a time horizon $T > 0$ and consider the partition $0 = t_0 < t_1 < \dots < t_N = T$ of $[0, T]$ with a (for simplicity assumed constant) step size $\Delta = \nicefrac{T}{N}$. Then, for smaller value of $\Delta$ resp. higher values of $N$, we need to increase $J$ to still obtain a similarly satisfactory approximation accuracy as on a coarser time grid. \cite{coutin2007approximation} study the accuracy and the convergence of the approximation of fractional Brownian motion via a finite number of OU processes taking also into account the effect of the temporal approximation. With their results in mind, we will make the following choices concerning the number of OU processes building the approximations \eqref{eq:finapproxWH} resp. \eqref{eq:finapprox}.
\begin{itemize}
\item[(i)] Given $N$ and $H < \nicefrac{1}{2}$, we take
\begin{equation}\label{eq:Jchoicei}
J \defeq J(N, H) = \floor{2 \cdot N^\zeta \cdot \log(N)} \, , \quad \zeta = \log(1 + H) \, ,
\end{equation}
so that $\zeta < H$ as $H \in (0, \nicefrac{1}{2})$.
\item[(ii)] If we do not want the dimension of the approximation to depend on $H$,\footnote{When implementing the nested particle filter of of Section~\ref{sub:NPF}, in order to be able to update the particles in the recursive step, the dimension of the OU approximation corresponding to a jittered parameter has to be the same as the dimension of the OU approximation corresponding to the `previous', unjittered parameter. This works quite well as, with our choice of jittering kernel, jittered parameters are obtained as small perturbations of the existing particles in the parameter space; however, it might be more consistent to keep the dimension fixed from the beginning and for all particles. This also allows us to work with tridimensional arrays in \textsf{R}.}
\begin{equation}\label{eq:Jchoiceii}
J \defeq J(N) = \floor{2 \cdot N^\zeta \cdot \log(N)} \, , \quad \zeta = \log(1 + 0.25) \, .
\end{equation}
\end{itemize}

\end{remark}
\subsection{Discrete-Time Modelling Framework}\label{sub:DTmodel}
In order to solve the estimation and inference problems resulting from our modelling setup, one possibility is to use the \textit{nested particle filtering algorithm} introduced by \cite{crisan2018nested}. In order to fit our setup within their framework, we must consider a time-discretized version of our model. To this, we fix a time horizon $T > 0$ and consider the partition $0 = t_0 < t_1 < \dots < t_N = T$ of $[0, T]$ with a (for simplicity assumed constant) step size $\Delta = \nicefrac{T}{N}$. We only consider changes in the unobservable state process at these discrete time points, which implies that the intensity has a piecewise-constant form. In particular, we assume
\begin{equation}\label{eq:pcint}
\lambda_u = b \cdot \exp\left(X_{t_{n - 1}}\right)
\end{equation}
for all $u \in (t_{n - 1}, t_n]$, $n = 1, \dots, N$. As long as $\Delta$ is small enough, this assumption is not too restrictive as the choice \eqref{eq:pcint} still allows to mimic true market behavior reasonably well.

We will use the following notation for the discretized model. We write $X_n = X_{t_{n - 1}}$ for the random sequence corresponding to the underlying state process, $\lambda_n$ for the corresponding intensity and, for what concerns observations, we consider the sequence given by $y_n = D_{t_n} - D_{t_{n - 1}}$, $n = 1, \dots, N$. Note that the piecewise-constant form of the intensity implies that $y_n$ is Poisson distributed; in particular we have
\begin{equation}\label{eq:poisln}
y_n \sim \text{Pois}(\lambda_n \Delta) \,.
\end{equation}
\subsubsection*{Updating Scheme}
To conclude this section, note that the particle filtering algorithms of Section~\ref{section:PF} require us to `update' the value of $J$ Ornstein-Uhlenbeck processes building the approximation \eqref{eq:finapprox}; that is, given their values at some time $t_{n - 1}$, we need to calculate their values at some later time $t_n$. Since each of these processes starts at and reverts to zero and has unit variance, the dynamics of each $Z^j$, $j = 1, \dots, J$ are given by
\[
dZ_t^j = - \kappa_j Z_t^j \, dt + dB_t \, .
\]
As explained in \cite{gillespie1996exact}, an Euler-Maruyama discretization of these dynamics is accurate only for a suitably small time-discretization step $\Delta$. In particular, $\Delta$ should be much smaller than the reciprocal of the mean reversion speed, which is clearly problematic in our context, given the high values of some of the $\kappa_j$'s in the approximation \eqref{eq:finapprox}. \cite{gillespie1996exact} explains how to circumvent this problem by using so-called \textit{exact updating formulae}, which allow us to calculate the value of $Z_{t_n}^j$ given $z \defeq Z_{t_{n - 1}}^j$ as
\begin{equation}\label{eq:updateOU}
Z_{t_n}^j = z e^{-\kappa_j \Delta} + \sqrt{\frac{1 - e^{-2\kappa_j \Delta}}{2\kappa_j}} \, v \,
\end{equation}
where $v$ is sampled from a standard normal distribution.

It is important to note that, since all the OU processes in \eqref{eq:finapprox} are driven by the same Brownian motion, we only need to generate a single standard normal random variable to update all of their values from time $t_{n-1}$ to time $t_n$, making for an efficient implementation of the particle filtering algorithms of Section~\ref{section:PF}.

Once again, we will stay within the discretized framework; therefore, in what follows we will consider the random sequence indexed by $n = 1, \dots N$ which we obtain by setting $Z_n^j = Z_{t_{n - 1}}^j$. 
\section{Nested Particle Filter}\label{section:PF}
The goal of this paper is not only to perform the filtering task, but also to estimate model parameters (in particular, $H$). To this, we rely on a \textit{nested particle filter} approach, as described in \cite{crisan2018nested}: the authors introduce a nested structure, employing two layers of particle filters, for the approximation of the joint posterior distribution of the signal and the unknown (static) parameters.

\subsubsection*{Main Assumptions for the Applicability of the Algorithm}
To determine the applicability of the algorithm in a specific framework, two main assumptions identified in \cite{crisan2018nested} must be satisfied. The first assumption states that the parameter space must be a compact set and that the conditional p.d.f. of the observations must be well-behaved (positive and upper bounded) uniformly over that set.  We will assume that the parameter space of our parameter of interest, the Hurst index $H$, is a compact subset of $(0, \nicefrac{1}{2})$, so that this assumption is satisfied.

Secondly and most importantly, the optimal filter for the model of interest must be continuous with respect to the parameter; that is, small changes in the parameter should lead only to small changes in the posterior of the state given the observations. This assumption is crucial to obtain a recursive algorithm. In fact, at a given time instance, the key quantity required to approximate the posterior measure of the parameter is its marginal likelihood. At time $n$, consider the sampled parameter $\bar{\theta}_n$ obtained, for instance, as a random mutation of $\theta_{n-1}$ through jittering: one could approximate the likelihood of $\bar{\theta}_n$ by running a standard bootstrap filter conditional on $\Theta = \bar{\theta}_n$ from time $0$ to time $n$. However, this non-recursive approach would dramatically increase the computational cost with time. Due to the continuity assumption, one can instead use the particle approximation of the filter corresponding to $\theta_{n-1}$ as a particle approximation of the filter corresponding to $\bar{\theta}_n$. This is crucial to approximate both the predictive measure and the likelihood of $\bar{\theta}_n$ by exploiting the `mismatched' parameter value $\theta_{n-1}$ instead, and it ensures the recursive property of the algorithm. We discuss this continuity property in our setting in Appendix~\ref{app} (see, in particular, Proposition~\ref{prop:cont}).

\subsection{`Inner' (Bootstrap) Filter}\label{sub:BF}
To fix notation and describe how the approximation of fBM can be exploited in the given context, in this section we review the standard particle filter (so-called \textit{bootstrap} filter); this represents the ``inner'' filter of the two nested layers of particle filters described in \cite{crisan2018nested} and it yields, at each time point $n$, an approximation $\pi_n$ of the posterior measure for the state, conditional on the observations up to the time point $n$ and the given (known or sampled) parameters.

We apply a bootstrap filter with $M$ particles in the context of the discrete-time version of our model, as described at the end of Section~\ref{sub:CTFram}. Regarding observations, we record the values of the observed realization of the counting process $D$ on a grid of step size $\Delta = \nicefrac{T}{N}$ to then obtain the observation sequence $(y_n)_{1 \leq n \leq N}$. Here, we consider both the constant $b$ and the Hurst index $H < \nicefrac{1}{2}$ to be known; then, using \eqref{eq:Jchoicei} we compute the number $J = J(H, N)$ of OU processes needed for the approximation. Finally, note that in the following we write $\bm{Z}_n \defeq (Z_n^1, \dots, Z_n^J)^T$.

\begin{enumerate}\setcounter{enumi}{0}
\item \textbf{Initialization} ($n = 0$)
\begin{itemize}
\item[-] Compute coefficients and mean-reversion speeds for the approximation \eqref{eq:finapprox}.
\item[-] For each $m = 1, \dots, M$, initialize: sample $\bm{Z}_0^m$ from a suitable prior distribution.\footnote{In our simulation studies, we will typically choose a centered $J$-dimensional normal distribution with a small variance.}
\end{itemize}
\item \textbf{Recursive step} (from $n - 1$ to $n$)

For $m = 1, \dots, M$, let $\{\bm{Z}_{n - 1}^m\}_{1 \leq m \leq M}$ be the particles (Monte Carlo samples) available at time $n - 1$. At time $n$:
\begin{itemize}
\item[(a)] \underline{Update}: for $m = 1, \dots, M$, generate $v^m \sim \mathcal{N}(0, 1)$ and draw $\bm{\bar{Z}}_n^{m}$ conditional on $\bm{Z}_{n - 1}^{m}$, see also \eqref{eq:updateOU}.
\item[(b)] \underline{Compute normalized weights}: for $m = 1, \dots, M$, compute (see \eqref{eq:finapprox})
\[
\bar{X}_n^m = \sum_{j = 1}^J c_j \bar{Z}_n^{j, m} \, ,
\]
the corresponding intensity $\lambda_n^m$ and the corresponding likelihood weight  (see also \eqref{eq:poisln})
\[
w_{t + \Delta}^m \propto (\lambda_n^m \Delta)^{y_n} \cdot \exp\left(-\lambda_n^m \Delta\right) \, .
 \]
Normalize the weights to obtain $\bar{w}_{t + \Delta}^m = w_{t + \Delta}^m / \sum_{m = 1}^M w_{t + \Delta}^m$.
\item[(c)] \underline{Resample}: For $m = 1, \dots, M$, let $\bm{Z}_n^m = \bar{\bm{Z}}_n^q$ with probability $\bar{w}_{t + \Delta}^q$, $q \in \{1, \dots, M\}$.
\end{itemize}
\end{enumerate}

\subsection{Nested Particle Filter and Parameter Estimation}\label{sub:NPF}
Consider the setting of the previous section, but let $\Theta$ be the vector of model parameters to be estimated. In our case, we aim at estimating $\Theta = H$. The estimation of $b$ will not be addressed in the specific context of this model, as it can typically be chosen to match the average number of price movements observed in one (continuous-time) unit. The nested particle filter for the approximation, at each discrete time point $n$, of the posterior distribution of the parameter is given below.

Other than the time step size $\Delta = \nicefrac{T}{N}$ and the recorded observations $(y_n)_{1 \leq n \leq N}$, the required inputs are the value of constant $b$ and the number of particles $K \cdot M$. For the parameter, an appropriate prior distribution has to be chosen: we use a uniform distribution over $(0, \nicefrac{1}{2})$.

\begin{enumerate}\setcounter{enumi}{0}
\item \textbf{Initialization} ($n = 0$)
\begin{itemize}
\item[-] Draw $K$ i.i.d. samples $\theta_0^k$, $k = 1, \dots, K$ from the prior distribution $\pi_0^\theta(d\theta)$. Then given the total number of time steps $N$, compute the corresponding $J^k = J(N, \theta_0^k)$, $k = 1, \dots, K$.
\item[-] For each $m = 1, \dots, M$ and for each $k = 1, \dots, K$, initialize: sample $\bm{Z}_0^{(k, m)}$ from a suitable prior distribution.
\end{itemize}
\item \textbf{Recursive step} (from $n - 1$ to $n$)

For $k = 1, \dots, K$, let $\left(\theta^k, \left\{\bm{Z}_n^{(k, m)}\right\}_{1 \leq m \leq M} \right)$ be the particle set available at time $n - 1$. At time $n$:
\begin{itemize}
\item[(a)] \underline{Jittering and state propagation}: for each $k = 1, \dots, K$, perform the following steps.
\begin{itemize}
\item[-] Draw $\bar{\theta}_n^k$ from $\kappa_K^{\theta_{n - 1}^k} (d\theta)$. This corresponds to \textit{jittering} the samples in the parameter space: for possible ways to perform such step, we refer to Section 4.2 of \cite{crisan2018nested}.
\item[-] For each $\bar{\theta}_n^k$, $k = 1, \dots, K$, compute the corresponding coefficients and mean-reversion speeds for the approximation \eqref{eq:finapprox}.
\item[-] Given $\bar{\theta}_n^k$, perform Step 2(a) of the standard particle filter described in the previous section to obtain $\bar{\bm{Z}}_n^{(k, m)}$ (and $\bar{X}_n^{(k, m)}$), $m = 1, \dots, M$.
\item[-] Compute the approximate likelihood of the parameter $\bar{\theta}_n^k$, given by
\[
u_n^M(\bar{\theta}_n^k) = \frac{1}{M} \sum_{m = 1}^M L_{\bar{\theta}_n^k} \left(\bar{\bm{Z}}_n^{(k, m)}\right)\,,
\]
where $L_{\bar{\theta}_n^k} \left(\bar{\bm{Z}}_n^{(k, m)}\right)$ correspond to the (Poisson) likelihood given $\bar{\theta}_n^k$.
\item[-] Perform steps 2(b) and 2(c) of the standard particle filter described in the previous section to update and obtain $\widetilde{\bm{Z}}_n^{(k, m)}$, $m = 1, \dots, M$.
\end{itemize}
\item[(b)] \underline{Compute normalized weights}: for each $k \in\{1, \dots, K\}$, compute normalized weights $\bar{w}_n^k \propto u_n^M(\bar{\theta}_n^k)$.
\item[(c)] \underline{Resample}: For $k = 1, \dots, K$, let
\[
\left(\theta_n^k, \left\{\bm{Z}_n^{(k, m)}\right\}_{1 \leq m \leq M} \right) = \left(\bar{\theta}_n^l, \left\{\widetilde{\bm{Z}}_n^{(l, m)}\right\}_{1 \leq m \leq M} \right)
\]
with probability $\bar{w}_n^l$, $l \in \{1, \dots, K\}$. Furthermore, for each resampled $\theta_n^k$, $k = 1, \dots, K$, we take the corresponding $J^k$.
\item[(d)] Approximate the posterior measure by $\mu_n \approx \frac{1}{K} \sum_{k = 1}^K \delta_{\theta_n^k}$.
\end{itemize}
\end{enumerate}

\section{Numerical Results}\label{section:NUM}
In this section, we present the results of a simulation study aimed at testing the accuracy of the presented algorithms in the context of our model, as well as an application of the proposed methodology to the ``alternative'' models of \cite{cont2022rough} and \cite{rogers2023things}.

\subsection{Simulation}\label{section:SIM} For the simulation experiments of this subsection, we fix a time horizon $T = 1$ and a small time step of size $\Delta = \nicefrac{1}{960}$; that is, we consider half-minute intervals in a trading day of $8$ hours. Except when analyzing the impact of the `informativeness' of the observation process, we set $b = 8000$ as the average amount of price changes in a given day -- a relatively conservative choice for a liquid stock.

\subsubsection*{Bootstrap Filter for Known $H$}\label{subsec:BF}
Here we consider the Hurst index $H$ to be known and thus we implement the bootstrap filter of Section~\ref{sub:BF} with $M = 600$ particles. We consider two cases for the true Hurst index, $H = 0.1$ and $H = 0.4$; given the total number of time steps $N = 960$, the number of OU processes building the approximation \eqref{eq:finapprox} is given by Equation~\eqref{eq:Jchoicei} as $J(960, 0.1) = 26$ and $J(960, 0.4) = 138$, respectively.

Figure~\ref{filt0104} shows simulated trajectories of Liouville Brownian motion (gray lines) with Hurst index $H = 0.1$ (upper panel) resp. $H = 0.4$ (lower panel), as well as the corresponding filtered estimates (black lines). We observe that the filtered trajectory tracks the true one reasonably well in both cases, particularly if we take into account the small time scale considered.

\begin{figure}[h]
\begin{center}
\centering
\captionsetup{justification=centering}
\includegraphics[scale=0.4]{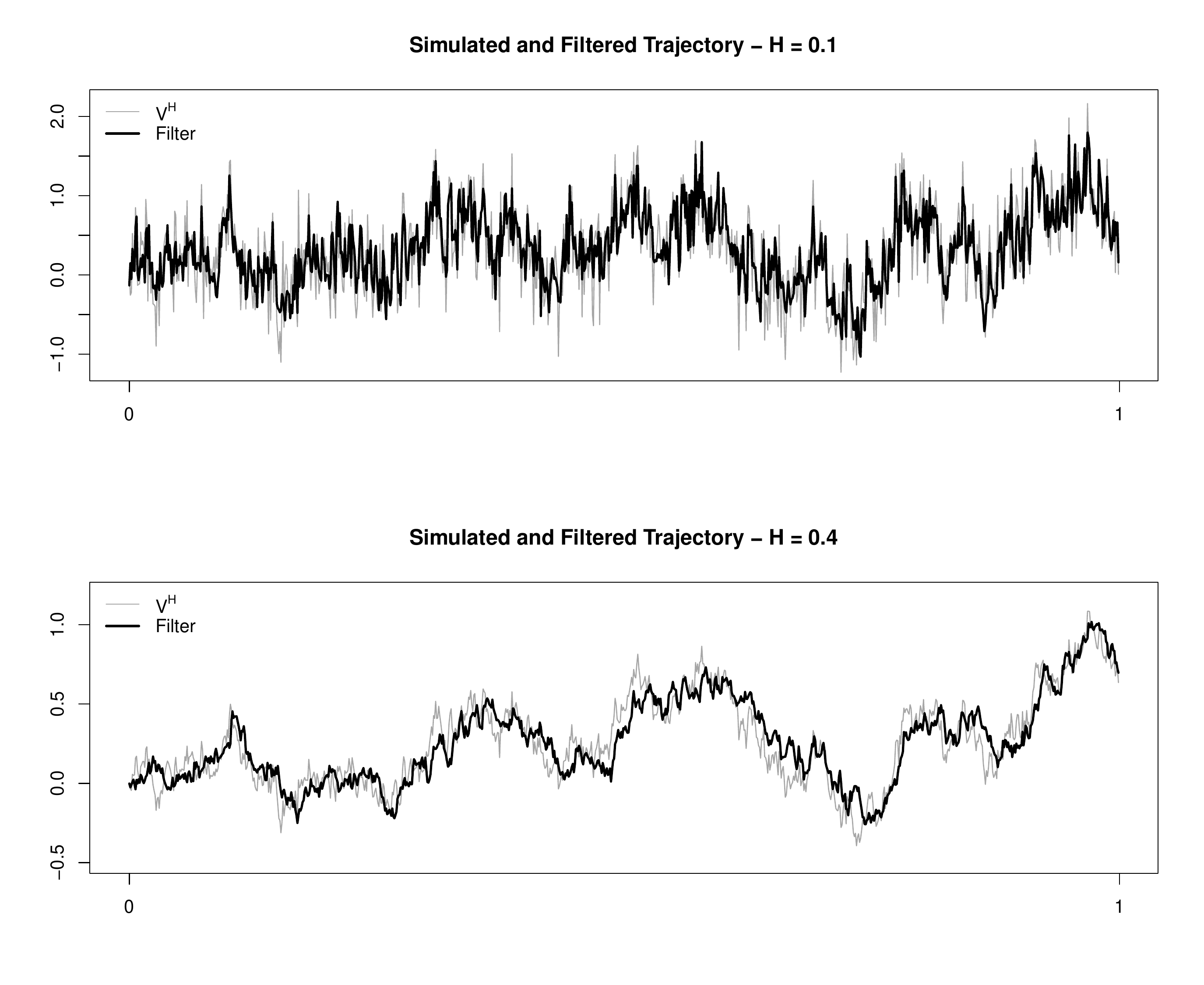}
\caption{True (gray) and filtered (black) state trajectory of $V^H$.}
\label{filt0104}
\end{center}
\end{figure}

\subsubsection*{Nested Particle Filter for Unknown $H$}\label{sec:simNPF}
Here we consider the Hurst index $H$ to be unknown and thus we estimate it accordingly by implementing the nested particle filter of Section~\ref{sub:NPF} with $K \cdot M = 300^2$ particles. We consider two cases for the true Hurst index, $H = 0.1$ and $H = 0.4$; however, unlike in the filtering experiment, we keep the number of OU processes building the approximation \eqref{eq:finapprox} equal in both cases. This number depends only on the total number of time steps $N = 960$, and it is given by Equation~\eqref{eq:Jchoiceii} as $J(960) = 63$.

Figure~\ref{mre0104} shows -- on a log scale -- the average, over 50 independent simulations, of the relative error $|\hat{H}_n - H|/H$, $n = 1, \dots, N$, in the case in which the true Hurst index $H$ is given by $H = 0.1$ (upper panel) resp. $H = 0.4$ (lower panel). We denote by $\hat{H}_n$ the mean of the posterior distribution $\mu_n$ obtained in a given run at the $n^{\text{th}}$ time step. Overall, the behavior of parameter estimates over time is quite satisfactory, as the estimation problem at hand is relatively complex (particularly when considering such a fine time grid and only one continuous-time unit).
\begin{figure}[h]
\begin{center}
\centering
\captionsetup{justification=centering}
\includegraphics[scale=0.4]{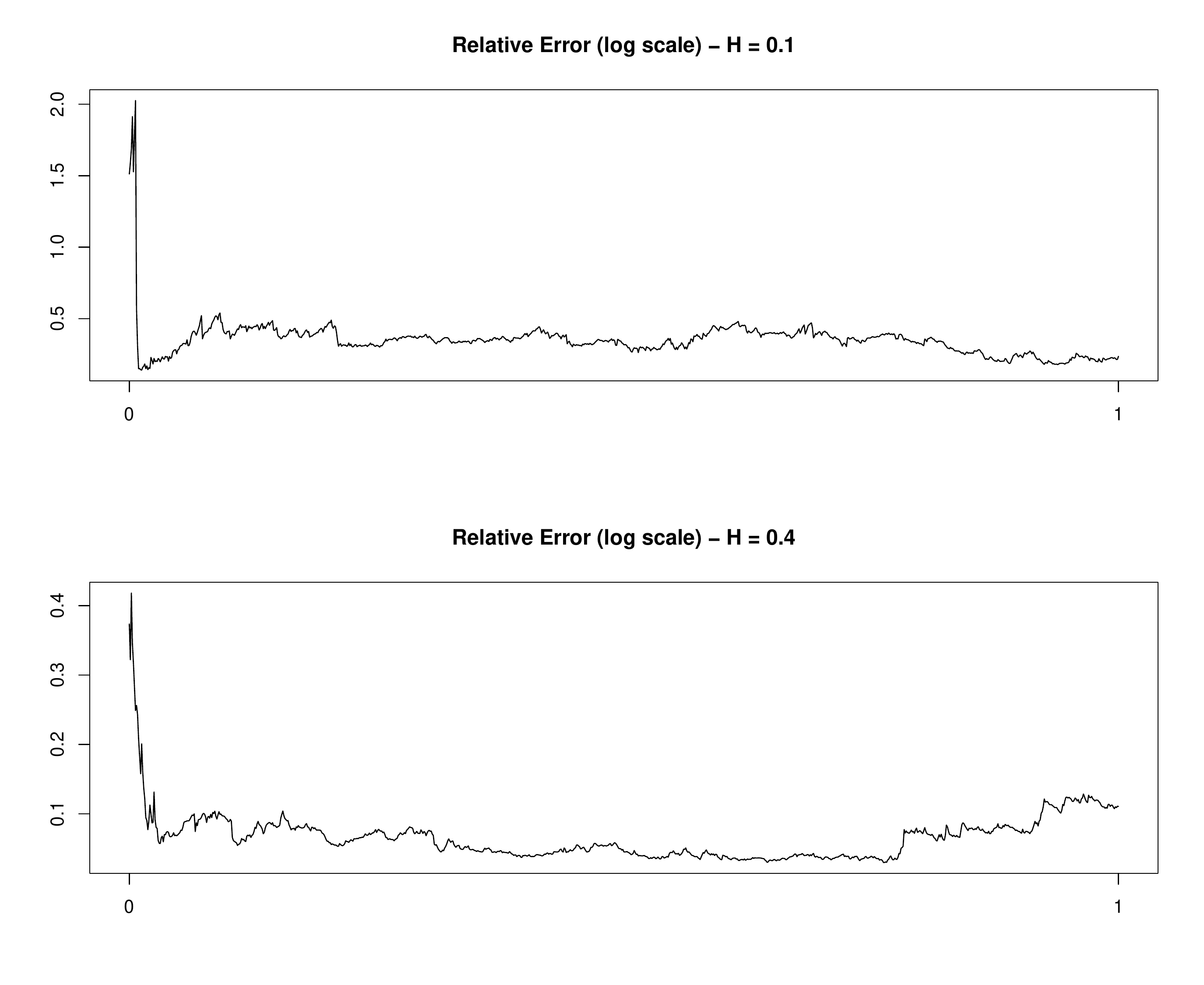}
\caption{Mean relative error (log scale) in parameter estimation.}
\label{mre0104}
\end{center}
\end{figure}
 
\subsubsection*{Sensitivity w.r.t. Value of $b$}
The biggest influence on the accuracy of the algorithms is that of the parameter $b$, which is why we focus on it next. Obviously, the number of particles ($M$ resp. $K \cdot M$) has also an effect, in the sense that the larger the number of particles, the better the accuracy of the algorithm. However, any reasonable, albeit quite conservative, choice would produce satisfactory filtering results.

Here we assume that the the true Hurst index is $H = 0.3$ and we analyze two cases, one in which the observation process is less informative (i.e., $b = 3000$) and one in which it is more informative (i.e., $b = 10000$). First, we consider the Hurst index to be known and thus we implement the bootstrap filter of Section~\ref{sub:BF} with $M = 300$ particles and -- as specified by Equation~\eqref{eq:Jchoicei} -- we use $J(0.3, 960) = 83$ OU processes to build the approximation \eqref{eq:finapprox}. Then, we consider $H$ to be unknown and we estimate it by implementing the nested particle filter of Section~\ref{sub:NPF} with $K \cdot M = 300^2$ particles, where the dimension of the approximation depends only on the total number of time steps $N$ and it is given by Equation~\eqref{eq:Jchoiceii} as $J(960) = 63$.

The upper panel of Figure~\ref{effectb} shows the unobserved trajectory of Liouville Brownian motion with Hurst index $H = 0.3$ (black solid line) and the area comprised between the $1\%$- and $99\%$-quantile of the posterior distribution of the state process. The lower panel of Figure~\ref{effectb} shows the posterior distribution, at final time, of the Hurst index in each case; the vertical lines indicate the respective means. In both plots, the lighter color corresponds to the less informative case, the darker one to the more informative one. In the lower panel, the parameter prior is also depicted (dotted line). We can observe how a higher value of $b$ improves the accuracy in both filtering and parameter estimation; however, results are still quite satisfactory in the less informative case. Note that, in order to make the modeling framework consistent with the large number of transactions characterizing high-frequency trading of liquid stocks, $b$ would typically take a large value in most real-world applications.
\begin{figure}[h]
\begin{center}
\centering
\captionsetup{justification=centering}
\includegraphics[scale=0.4]{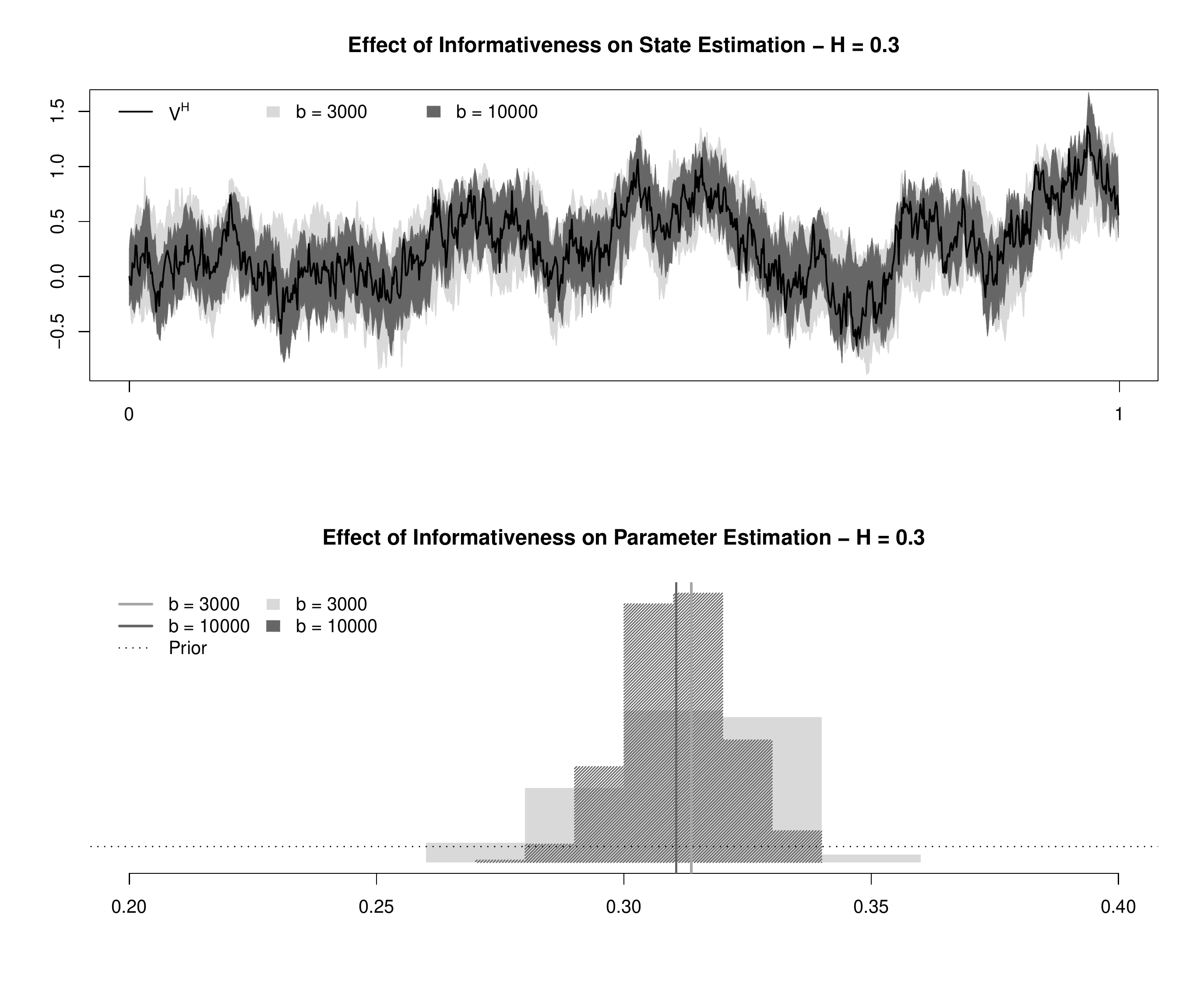}
\caption{True and filtered state trajectory of $V^H$ with $H = 0.3$ (upper panel) and posterior distribution of $H$ at final time (lower panel) in the cases $b = 10000$ resp. $b = 3000$.}
\label{effectb}
\end{center}
\end{figure}

\subsection{Experiments in the Context of Non-Rough Stochastic Volatility Models}\label{sec:norough}

In this section, in the spirit of \cite{rogers2023things} and \cite{cont2022rough}, we apply our methodology to synthetic data generated in a setting where (spot) volatility is driven by standard Brownian diffusions and assess whether the resulting estimates for $H$ are more consistent with such a setup, rather than mistakenly reflecting a spurious strong roughness effect. This amounts to modify the continuous-time model of Section~\ref{sub:CTFram}, discretize it in the same fashion as in Section~\ref{sub:DTmodel} and estimate the unobservable volatility trajectory and the Hurst index $H$ using the nested particle filter algorithm of Section~\ref{sub:NPF} using $K \cdot M = 300^2$ particles.

\subsubsection*{Volatility as Modulus of Brownian Motion} To be consistent with Section~4.1 of \cite{cont2022rough}, here we assume our observation process $D$ to have an intensity of the form (to be compared with \eqref{eq:int}):
\begin{equation}\label{eq:intnorough}
\lambda_t \defeq \lambda(W_t) = b \cdot |W_t|^2 \, ,
\end{equation}
where $b$ is a positive constant and $W$ is a standard Brownian motion, i.e. here $H = \nicefrac{1}{2}$. However, we consider the Hurst index $H$ to be unknown and thus we estimate it accordingly by implementing the nested particle filter of Section~\ref{sub:NPF}. We fix a time step of size $\Delta = \nicefrac{1}{480}$ (that is, we consider one-minute intervals in a trading day of $8$ hours) and a time horizon of $T = 5$ days, so that $N = 2400$. Given the total number of time steps $N$, Equation~\eqref{eq:Jchoiceii} gives $J(2400) = 88$ for the number of OU processes building the approximation \eqref{eq:finapprox}.

It is very important to stress that this experiment aims at assessing whether the resulting estimate for $H$ reflects a spurious strong roughness effect, attributable for instance to microstructure noise, rather than at estimating $H$ correctly. In fact, the value $H = \nicefrac{1}{2}$ can \textit{never} be estimated correctly without modifying the methodology and the algorithm substantially, as the approximations \eqref{eq:finapproxWH} and \eqref{eq:finapprox} are valid for $H < \nicefrac{1}{2}$ and the parameter space of nested particle filtering algorithm of Section~\ref{sub:NPF} is assumed to be a compact subset of $(0, \nicefrac{1}{2})$.

The upper panel of Figure~\ref{modulus} shows a simulated trajectory of such $|W|$ (gray line) and its corresponding filtered estimate (black line). The lower panel shows the mean (dark gray line) and the $1\%$- resp. $99\%$- quantiles (light gray lines) of the estimated posterior $\mu_n$, $n = 1, \dots, N$, of the unknown parameter $H$, when the underlying unobservable process is (the modulus of) a standard Brownian motion, i.e. $H = \nicefrac{1}{2}$.
\begin{figure}[h]
\begin{center}
\centering
\captionsetup{justification=centering}
\includegraphics[scale=0.4]{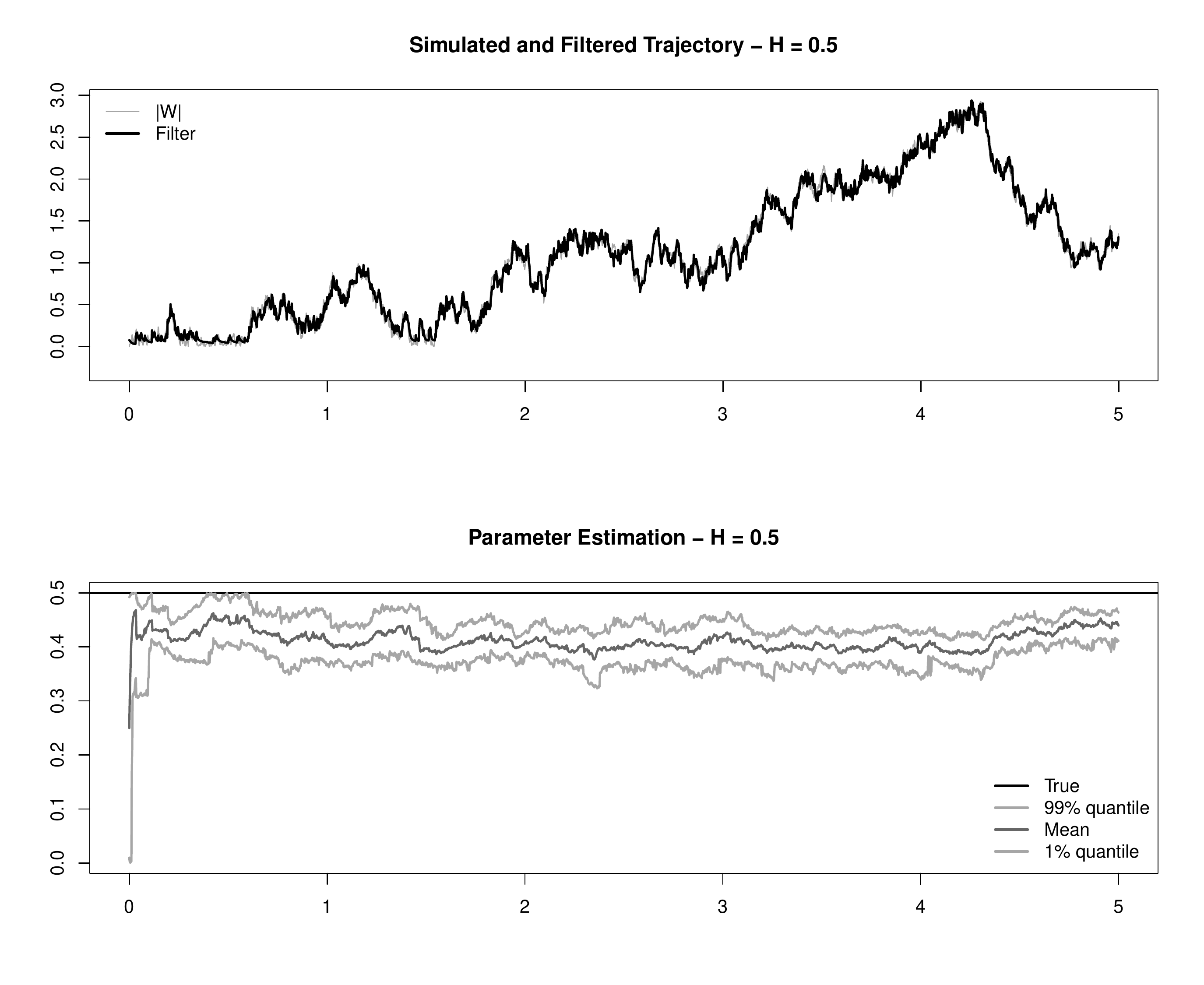}
\caption{True and filtered state trajectory of $|W|$ and behavior of parameter estimates over time.}
\label{modulus}
\end{center}
\end{figure}

Overall, considering that $H = \nicefrac{1}{2}$ can never be estimated in our current setup, the results are quite satisfactory, both in terms of filtered trajectory and -- especially -- in terms of parameter estimate, since the latter is not falsely identified to be close to zero, but rather it nears $\nicefrac{1}{2}$.

\subsubsection*{OU-OU Model for Volatility} To be consistent with Section~4 of \cite{rogers2023things}, here we assume our observation process $D$ to have an intensity of the form (to be compared with \eqref{eq:int}):
\begin{equation}\label{eq:intOUOU}
\lambda_t \defeq \lambda(V_t) = b \cdot V_t^2 \, ,
\end{equation}
where $V$ is given by the so-called OU-OU model:
\begin{equation}\label{eq:VOUOU}
\begin{split}
dR_t =  -  \beta R_t \, dt + \sigma_R \, dW_t' \, , \\
dV_t =  \kappa (R_t - V_t) \, dt + \sigma_V \, dW_t \, .
\end{split}
\end{equation}
To simulate model and observations, we choose the same parameters as in the original paper; that is, $\sigma_V^2 = 20$, $\sigma_R^2 = 0.625$, $\kappa = 210$ and $\beta = 2.5$. However, when implementing our estimation methodology, we suppose that the unobservable process $V$ is not of its `true' form \eqref{eq:VOUOU}, but rather a Liouville Brownian motion with unknown Hurst index $H$, which we estimate using the nested particle filter of Section~\ref{sub:NPF}. We fix a time horizon of $T = 5$ days and make two distinct choices for the time step $\Delta$ to analyze its impact on the estimation procedure. In particular, we consider $\Delta = \nicefrac{1}{240}$ (corresponding to two-minute intervals in a trading day of $8$ hours) and $\Delta = \nicefrac{1}{960}$ (corresponding to thirty-second intervals in a trading day of $8$ hours); thus, the total number of steps is $N = 1200$ resp. $N = 4800$ and Equation~\eqref{eq:Jchoiceii} gives $J(1200) = 68$ resp. $J(4800) = 112$ for the number of OU processes building the approximation \eqref{eq:finapprox}.

Experiments show that here the estimates for $H$ are more consistent with a rough volatility model, meaning that the proposed OU-OU model seem to mimic that type of behavior. However, if we simulate the model on a finer time grid and apply the corresponding algorithm, the estimate for the Hurst index increases. The upper panel of Figure~\ref{rogers}, in particular, shows the posterior distribution, at final time, of the Hurst index estimated from the OU-OU model in the case $\Delta = \nicefrac{1}{240}$ and in the case $\Delta = \nicefrac{1}{960}$; the vertical lines indicate the respective means. The lower panel shows the corresponding posterior distributions for a model of intensity of the form \eqref{eq:intOUOU} above, i.e. $\lambda_t \defeq \lambda(V_t^H) = b \cdot \left(V_t^H\right)^2$, where the driving process $V^H$ is of the form specified in Equation~\eqref{eq:LiouBM} with $H = 0.2$ (roughly corresponding to the average of the estimates obtained for the OU-OU model on the two different time grids). In both plots, the lighter color corresponds to the first case, the darker one to the more finely discretized one and the dotted line to the parameter prior. We can observe how Hurst index estimates resulting from applying the algorithm to observations generated from the OU-OU model vary strongly depending on the fineness of the chosen time scale (in particular, estimates increases -- and thus mimic less roughness -- the finer the time grid is), while this is not the case when observations are generated from the `true' model.

\begin{figure}[h]
\begin{center}
\centering
\captionsetup{justification=centering}
\includegraphics[scale=0.4]{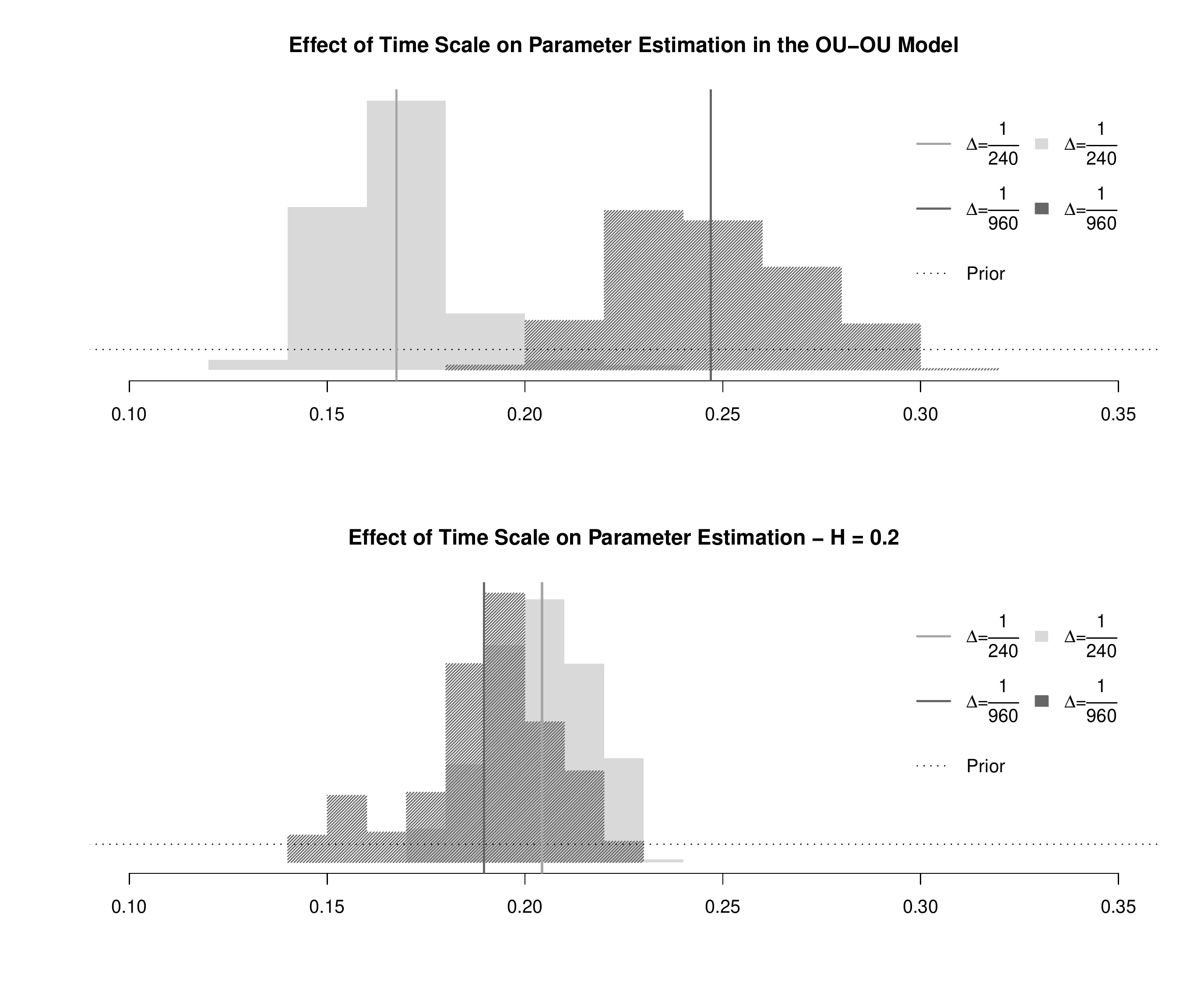}
\caption{Posterior distribution of $H$ at final time in the cases $\Delta = \nicefrac{1}{240}$ and $\Delta = \nicefrac{1}{960}$.}
\label{rogers}
\end{center}
\end{figure}

\section{Conclusions and Outlook}\label{section:concl}
In this paper, we discuss filtering and parameter estimation in a rough volatility model for high-frequency data. We consider a time-discretized version of a continuous-time framework where observations are given by a trajectory of a Cox process whose intensity is driven by a `rough' process, such as fractional resp. Liouville Brownian motion with Hurst parameter $H < \nicefrac{1}{2}$. Basing ourselves on a representation of such signal as a superposition of Ornstein Uhlenbeck processes (as introduced by \cite{carmona1998fractional}), we detail how it can be filtered using standard particle filtering techniques. Moreover, we explain how to estimate $H$ using the \textit{nested particle filtering} algorithm of \cite{crisan2018nested}. We run a comprehensive simulation study to test the accuracy of the algorithms and their sensitivity with respect to specific modelling choices. We find that the results are satisfactory; in particular, we are able to estimate the Hurst index with reasonable accuracy, while still taking the unobservability of volatility fully into account. Moreover, the parameter estimation methodology seems to be able to adequately distinguish between `true' rough dynamics in spot volatility and spurious roughness effects arising from microstructure noise, as documented e.g. in \cite{cont2022rough}.

From an implementation point of view, a possible extension involves modifying the algorithms so that they do not perform resampling at every step as they currently do, but only when the diversity in the particles is reduced beyond a certain threshold. This is typically done by estimating the so-called Effective Sample Size (ESS), although in the setting of the \textit{nested particle filtering} algorithm of \cite{crisan2018nested} the usual estimator of ESS should be modified to accommodate the special characteristics of the algorithm and to avoid becoming uninformative (for further reading about the \textit{normalized} ESS estimator devised by \cite{crisan2018nested}, we refer the reader to Section~5.6 of that paper).

Moreover, the proposed methodology should be applied to real data. It would be particularly interesting to compare estimates of $H$ obtained using this methodology to those obtained from \cite{volrough}, \cite{bennedsen2021decoupling}, \cite{fukasawa2019volatility} and \cite{cont2022rough}, among others. However, as briefly touched upon in Section~\ref{sec:IntrfBM}, the \textit{long-memory} feature of volatility has been investigated in several empirical studies (see, for instance, \cite{andersen2001distribution}) and it is often considered a stylized fact. In the modelling framework considered so far, it is not possible to capture \textit{persistence}: as explained in \cite{gneiting2004stochastic}, models based on processes with the self-similarity property (such as fractional Brownian motion) cannot account for both roughness and long-run dependence. Moreover, neither fractional nor Liouville Brownian motion are stationary, which represents a problem when one seeks to study the long-term properties of such processes. Therefore, a possible extension of the present work is to find an alternative model that allows for rough paths and for a polynomial decay of the autocorrelation function $\rho(s)$ for $s \to \infty$ (that is, for long memory). This is very much in the spirit of \cite{bennedsen2021decoupling}, who advocate the use of the so-called Brownian semistationary processes (see \cite{barndorff2009brownian}) to model volatility. Preliminary results indicate that it is possible to modify the previously-considered modelling framework in a way such that it accommodates both effects, while preserving a structure suitable for the filtering and estimation strategies presented in Section~\ref{section:PF}; however, a detailed analysis and a comprehensive simulation study for this alternative model are topics of future research.

Finally, it is important to note that models where underlying random factors are driven by fractional Brownian motion do not only find their applications in financial settings: among several examples, we highlight the recent work by \cite{alos2020fractional}, which focuses on the COVID-19 pandemic, and that by \cite{leppanen2021sailing}, which concerns the analysis of functional magnetic resonance imaging (fMRI) scans. Therefore, the methodology presented in this paper might be applied to estimate fractional models in a variety of contexts beyond the present one.

\appendix
\section{Further Results}\label{app}
\subsection{Relationship to Stochastic Volatility}
\begin{proposition}\label{prop:intvol}
Consider a sequence of models indexed by $p$ and given by $S_t^p = S_0 + \sum_{i = 1}^{D_t^p} \nu_i^p$, where for all $p$
\begin{itemize}
\item[(i)] $D^p$ is a doubly-stochastic Poisson process with intensity $\lambda^p(X_t)$;
\item[(ii)] $\{\nu_i\}_{i \in \mathbb{N}}$ are i.i.d. with zero mean and variance given by $(\sigma^p)^2$;
\item[(iii)] $(\nu_i^p)^2 \leq \bar{c}_p$ for a sequence $\bar{c}_p \to 0$;
\item[(iv)] $(\sigma^p)^2 \lambda^p(x) \to \sigma^2 \lambda(x)$, uniformly on compacts.
\end{itemize}
Then the pair $(S^p, X)$ converges in distribution to $(S, X)$, where $S$ solves
\[
dS_t = \sigma \sqrt{\lambda(X_t)} \, dW_t
\]
for a standard Brownian motion $W$.
\end{proposition}
\begin{proof}[Sketch of proof.] We first condition on $\mathcal{F}_\infty^X = \sigma(X_s, s \geq 0)$, so that $D^p$ is a Poisson process with time-dependent intensity. Thus, consider the filtration $\{\mathcal{F}_t^S \vee \mathcal{F}_\infty^X\}_{t \geq 0}$. Note that, for each $p$, $S^p$ is a martingale on this filtration. Moreover, it holds that
\[
d(S_t^p)^2 = 2 S_{t^{-}}^p \, dS_t^p + d[S^p]_t \,,
\]
where $[S^p]_t = \sum_{i = 1}^{D_t^p} (\nu_i^p)^2$. That is, $[S^p]_t - \int_0^t (\sigma^p)^2 \lambda^p(X_s) \, ds$ is a martingale. It follows that $(S_t^p)^2 - \int_0^t (\sigma^p)^2 \lambda^p(X_s) \, ds$ is also a martingale.

By \textit{(iv)}, one has that $\int_0^t (\sigma^p)^2 \lambda^p(X_s) \, ds \to  \int_0^t \sigma^2 \lambda(X_s) \, ds$. Moreover, for fixed $T$, \textit{(iii)} gives
\[
\sup_{t \leq T} \left(S_t^p - S_{t^{-}}^p\right)^2 = \sup_{i \leq D_T^p} (\nu_i^p)^2 \leq \bar{c}_p \to 0 \, .
\]

Hence, Theorem 7.4.1 in \cite{ethier1986markov} gives that the law of $S^p$ given $\mathcal{F}_\infty^X$ converges to the law of the process $S$ with $dS_t = \sigma \sqrt{\lambda(X_t)} \, dW_t$. The claim then follows by integrating with respect to the law of $X$.
\end{proof}

Note that above we assume a linear model for the stock price itself, which could theoretically lead to negative prices; however, when considering a short time horizon this possibility can safely be ignored. When analyzing the model on a longer time horizon, the linear model should rather be used for the log price.

\subsection{Continuity of the filter in the Hurst index} Recall that, in our setup, the state process is in fact a functional of $J$ (time-discretized) OU processes driven by the same Brownian motion. In what follows we write $\bm{Z}_n \defeq (Z_n^1, \dots, Z_n^J)^T$, where for $n = 1, \dots N$ and $j = 1, \dots J$ we have set $Z_n^j = Z_{t_{n - 1}}^j$ (as described in Section~\ref{sub:DTmodel}). For fixed $H$, we consider mean-reversion speeds $\kappa_j \defeq \kappa_j(H)$ and coefficients $c_j \defeq c_j(H)$, $j = 1, \dots, J$, and we assume that they are continuous in $H$. Moreover, we let
\[
\lambda_n \defeq \lambda(\bm{Z}_n, H) = b \exp\left(\sum_{j = 1}^J c_j Z_n^j\right) \, .
\]

As for the observations, $y_n$ denotes the number of jumps over $(t_{n - 1}, t_n]$, $n = 1, \dots, N$. Thus, in our setting, $y_1, \dots, y_n$ are conditionally independent given $\bm{Z}_1, \dots, \bm{Z}_n$. Then, since for $n = 1, \dots, N$, $y_n \sim \text{Pois}(\lambda(\bm{Z}_n, H) \Delta)$, we can define the likelihood
\begin{equation}\label{eq:likPois}
g(y \vert z, H) = \exp\left(- \lambda(z, H) \Delta \right) \frac{(\lambda(z, H) \Delta)^y}{y!} \, .
\end{equation}

Now we can state the following continuity result, which is weaker than the Lipschitz continuity of \cite{crisan2018nested}, but nonetheless useful in our context.

\begin{proposition}\label{prop:cont}
Consider a sequence $H_h \to H$, $h \to \infty$ and fix a discrete time point $n \leq N$. Then, for $f: \mathbb{R}^J \to \mathbb{R}$ bounded and continuous, we have
\begin{equation}\label{eq:contparfilter}
\lim_{h \to \infty} \mathbb{E}\left[f(\bm{Z}_n) \vert y_1, \dots, y_n; H_h \right] = \mathbb{E}\left[f(\bm{Z}_n) \vert y_1, \dots, y_n; H \right] \, .
\end{equation}
\end{proposition}

\begin{proof}
First note that, for generic $H$, we have the following formula (of Kallianpur-Striebel type):
\[
\mathbb{E}\left[f(\bm{Z}_n) \vert y_1, \dots, y_n; H \right] = \frac{\mathbb{E}\left[f(\bm{Z}_n) \prod_{i = 1}^n g(y_i \vert \bm{Z}_i, H) \vert H \right]}{\mathbb{E}\left[\prod_{i = 1}^n g(y_i \vert \bm{Z}_i, H) \vert H\right]} \, .
\]

To establish convergence, it suffices to consider the numerator. Define the set
\[
\mathcal{Z}^L = \left\{\bm{Z}_1, \dots, \bm{Z}_N : \bm{Z}_n \in \mathbb{R}^J, \Vert \bm{Z}_n \Vert \leq L \quad \forall \, n = 1, \dots, N \right\} \, .
\]
As the $\bm{Z}_n$ are Gaussian with bounded variance, for fixed $y_1, \dots, y_N$ it holds that
\[
\lim_{L \to \infty} \sup_h \mathbb{E}\left[\prod_{i = 1}^n g(y_i \vert \bm{Z}_i, H_h) \, ; \, \left(\mathcal{Z}^L\right)^\mathsf{c}  \vert H_h \right] = 0 \, .
\]

Now, since $\kappa_j$, $j = 1, \dots, J$ are continuous in $H$, the law of $\bm{Z}_1, \dots, \bm{Z}_N$ given $H_h$ converges to the law of $\bm{Z}_1, \dots, \bm{Z}_N$ given $H$. Moreover, for the intensity it holds that $\lambda(z, H_h) \to \lambda(z, H)$ locally uniformly for $h \to \infty$. This gives
\[
\lim_{h \to \infty} \mathbb{E}\left[f(\bm{Z}_n) \prod_{i = 1}^n g(y_i \vert \bm{Z}_i, H_h)  \vert H_h\right] = \mathbb{E}\left[f(\bm{Z}_n) \prod_{i = 1}^n g(y_i \vert \bm{Z}_i, H)  \vert H \right]
\]
by standard arguments.
\end{proof}

\bibliographystyle{plainnat}
\addcontentsline{toc}{chapter}{References}
\bibliography{TeX}
\end{document}